\documentclass{llncs}

\usepackage{amsmath}
\usepackage{amssymb}
\usepackage{enumerate}
\usepackage{varioref}

\DeclareMathOperator{\FPT}{\mathrm{FTP}}
\DeclareMathOperator{\NP}{\mathrm{NP}}

\DeclareMathOperator{\BFS}{\mathrm{BFS}}

\newcommand{\es}{\varnothing}

\title{\sc {Rainbow domination and related problems on some classes of perfect graphs}} 

\author{
Wing-Kai~Hon\inst{1} 
\and 
Ton~Kloks 
\and 
Hsian-Hsuan~Liu\inst{1} 
\and
Hung-Lung~Wang\inst{2}
}

\institute{
Department of Computer Science\\
National Tsing Hua University, Taiwan
\and
Institute of Information and Decision Sciences\\
National Taipei University of Business, Taiwan\\
{\tt hlwang@ntub.edu.tw}} 

\begin{document}

\maketitle 

\begin{abstract}
Let $k \in \mathbb{N}$ and let $G$ be a graph. 
A function $f: V(G) \rightarrow 2^{[k]}$ is a rainbow  
function if, for every vertex $x$ with $f(x)=\es$, 
$f(N(x)) =[k]$. The rainbow domination 
number $\gamma_{kr}(G)$ is the minimum of $\sum_{x \in V(G)} |f(x)|$ 
over all rainbow functions. 
We investigate the rainbow domination 
problem for 
some classes of perfect graphs. 
\end{abstract}

\section{Introduction}

Bre\u{s}ar et al. introduced the rainbow domination problem in 
2008~\cite{kn:bresar}. 

The $k$-rainbow domination-problem drew 
our attention because 
it is solvable in polynomial time for classes of graphs of bounded rankwidth 
but, unless one fixes $k$ as a constant, it seems not formulatable  
in monadic second-order logic.  

\smallskip 

Let us start with the definition. 
 
\begin{definition}
Let $k \in \mathbb{N}$ and let $G$ be a graph. 
A function $f: V(G) \rightarrow 2^{[k]}$ is a $k$-rainbow 
function if, for every $x \in V(G)$, 
\begin{equation}
\label{eqn1}
f(x)=\es \quad \text{implies}\quad \cup_{y \in N(x)} \; f(y)\; = [k].
\end{equation}
The $k$-rainbow domination number of $G$ is 
\begin{multline}
\label{eqn2}
\gamma_{rk}(G) = \min \; \bigl\{\; \|f\| \; \mid \;   
\text{$f$ is a $k$-rainbow function for $G$}\;\bigr \},\\ 
\text{where} \quad \|f\|=\sum\nolimits_{x \in V(G)} \; |f(x)|. 
\end{multline} 
\end{definition}
We call $\|f\|$ the \underline{cost} of $f$ over the graph $G$. 
When there is danger of confusion, we  write $\|f\|_{G}$ 
instead of $\|f\|$.  We call the elements of $[k]$ the \underline{colors} 
of the 
rainbow and, for a vertex $x$ we call $f(x)$ the \underline{label} of $x$.   
For a set $S$ of vertices we write 
\[f(S)=\cup_{x \in S}\; f(x).\]  
 
\bigskip 

It is a common phenomenon that the introduction of a new domination 
variant is followed chop-chop by an explosion of research results and 
their write-ups. 
One reason for the popularity of domination problems is the wide 
range of applicability and directions of possible research. 
We moved our bibliography of recent publications 
on this specific domination variant to the appendix. 
We refer to~\cite{kn:sumenjak} for the description of an application 
of rainbow domination. 
 
\bigskip 

To begin with, Bre\u{s}ar et al. showed 
that, for any graph $G$,  
\begin{equation}
\label{eqn3}
\gamma_{rk}(G)=\gamma(G \Box K_k),
\end{equation}
where $\gamma$ denotes the domination number and where $\Box$ 
denotes the Cartesian product. This observation, together with  
Vizing's conjecture, stimulated the search for graphs for which 
$\gamma=\gamma_{r2}$ (see also~\cite{kn:aharoni,kn:hartnell}). 
Notice that, by~\eqref{eqn3} and Vizing's upperbound 
$\gamma_{rk}(G) \leq k \cdot \gamma(G)$~\cite{kn:vizing}. 

\bigskip 

Chang et al.~\cite{kn:chang} were quick on the uptake and showed that, 
for $k \in \mathbb{N}$,  
the $k$-rainbow domination problem is NP-complete, even when restricted to 
chordal graphs or bipartite graphs. The same paper shows 
that there is a linear-time algorithm to determine the parameter on trees. 
A similar algorithm for trees appears in~\cite{kn:yen} and this paper 
also shows that the problem remains NP-complete on planar graphs. 

\bigskip 

Notice that~\eqref{eqn3} shows that $\gamma_{rk}(G)$ 
is a non-decreasing function in $k$.  Chang et al. show that, 
for all graphs $G$ with $n$ vertices and all $k \in \mathbb{N}$,   
\begin{equation}
\label{eqn13}
\min \; \{\;k,\;n\;\} \leq \gamma_{rk}(G) \leq n \quad\text{and} \quad 
\gamma_{rn}(G)=n. 
\end{equation}
For trees $T$, Chang et al.~\cite{kn:chang} give sharp 
bounds for the smallest $k$ satisfying   
$\gamma_{rk}(T)=|V(T)|$. 

\bigskip 
 
Many other papers establish bounds and relations, eg,  
between the $2$-rainbow domination number 
and the total domination number  
or the (weak) roman domination 
number~\cite{kn:chellali,kn:fujita3,kn:furuya,kn:wu,kn:wu2}, 
or study edge- or vertex critical graphs with respect to 
rainbow domination~\cite{kn:rad}, or obtain results for special 
graphs such as paths, cycles, graphs with given radius, and 
the generalized Petersen 
graphs~\cite{kn:ali,kn:fujita4,kn:shao,kn:stepien,kn:stepien2,kn:tong,%
kn:wang,kn:xavier,kn:xu}. 

\bigskip 

Pai and Chiu develop   
an exact algorithm and a heuristic for $3$-rainbow domination. 
In~\cite{kn:pai} they present the results of some experiments.  
Let us mention that the $k$-rainbow domination number may be 
computed, via~\eqref{eqn3}, 
by an exact, exponential algorithm that computes the domination 
number. For example, this shows that the $k$-rainbow domination number 
can be computed in $O(1.4969^{nk})$~\cite{kn:rooij2,kn:rooij}. 

\bigskip 

A $k$-rainbow family is a set of $k$-rainbow functions which 
sum to at most $k$ for each vertex. The $k$-rainbow domatic number 
is defined as the maximal number of elements in such a 
family. Some results were obtained 
in~\cite{kn:fujita,kn:meierling,kn:sheikholeslami}. 

\bigskip 

Whenever domination problems are 
under investigation, the class of strongly chordal graphs are of interest 
from a computational point of view. 
Farber showed that a minimum weight dominating set 
can be computed in polynomial time on strongly 
chordal graphs~\cite{kn:farber}. 
Recently, Chang et al. showed that the $k$-rainbow dominating number 
is equal to the so-called weak $\{k\}$-domination 
number for strongly chordal 
graphs~\cite{kn:bresar,kn:bresar2,kn:chang2}.  
A weak $\{k\}$-dominating function 
is a function $g:V(G) \rightarrow \{0,\dots,k\}$ such that, 
for every vertex $x$, 
\begin{equation}
\label{eqn4}
g(x)=0 \quad \text{implies} \quad \sum_{y \in N(x)} g(y) \geq k.
\end{equation} 
The weak domination number $\gamma_{wk}(G)$ minimizes 
$\sum_{x \in V(G)} g(x)$, over all weak $\{k\}$-dominating functions $g$.  
In their paper, Chang et al. show that the $k$-rainbow domination number 
is polynomial for block graphs. As far as we know, the $k$-rainbow 
domination number is open for strongly chordal graphs. 

\bigskip 

It is easy to see that, for each $k$, the $k$-rainbow domination 
problem can be formulated in monadic second-order logic. This shows 
that, for each $k$, the parameter is computable in linear time 
for graphs of bounded treewidth or rankwidth~\cite{kn:courcelle}. 

\begin{theorem}
\label{thm courcelle}
Let $k \in \mathbb{N}$. There exists a linear-time algorithm that 
computes $\gamma_{rk}(G)$ for graphs of bounded rankwidth. 
\end{theorem}

For example, Theorem~\ref{thm courcelle} implies that, for each $k$, 
$\gamma_{rk}(G)$ is computable in polynomial time for 
distance-hereditary graphs, ie, the graphs of rankwidth 1. Also, 
graphs of bounded outerplanarity have bounded treewidth, which implies 
bounded rankwidth. 

\bigskip 
 
A direct application of the monadic second-order 
theory involves a constant which is an exponential function of $k$. 
In the following section we show that, often, this 
exponential factor can be avoided. 

\section{$k$-Rainbow domination on cographs}

Cographs are the graphs without an induced $P_4$. 
As a consequence, cographs are completely decomposable by 
series and parallel operations, that is, joins and unions~\cite{kn:gallai}. 
In other words, a graph is a cograph 
if and only if every nontrivial, induced subgraph is disconnected 
or its complement is disconnected. Cographs have a rooted, 
binary decomposition tree, called a cotree, 
with internal nodes labeled as joins 
and unions~\cite{kn:corneil3}.   

\bigskip 

For a graph $G$ and $k \in \mathbb{N}$, 
let $F(G,k)$ denote the set of $k$-rainbow functions on $G$. 
Furthermore, define 
\begin{eqnarray}
F^{+}(G,k) &=&  \{\; f \in F(G,k) \;\mid\; 
\forall_{x \in V(G)} \; f(x) \neq \es \; \} \\
\text{and} \quad F^{-}(G,k) & = & F(G,k) \setminus F^{+}(G,k).
\end{eqnarray}

\bigskip 

\begin{theorem}
\label{thm cograph}
There exists a linear-time algorithm to compute the $k$-rainbow domination 
number $\gamma_{rk}(G)$ for cographs $G$ and $k \in \mathbb{N}$. 
\end{theorem}
\begin{proof}
We describe a dynamic programming algorithm 
to compute the $k$-rainbow 
domination number.
A minimizing $k$-rainbow function can be 
obtained by backtracking.

\medskip

\noindent
It is easy to determine the minimal cost of $k$-rainbow functions that have 
no empty set-labels. We therefore concentrate on those $k$-rainbow 
functions for which some labels are empty sets. 

\medskip 
  
\noindent
Let $k \in \mathbb{N}$.  For a cograph $H$ define 
\begin{eqnarray}
\label{eqn5}
R^{+}(H) & =& 
\min \; \{\; \|f\|_{H} \; \mid \;  
f \in F^{+}(H,k) \; \text{and}\;  
f(V(H)) = [k] \;\}, \\
\label{eqn52}
R^{-}(H) &=& 
\min \; \{\; \|f\|_{H} \; \mid \; 
f \in F^{-}(H,k)\;\}.
\end{eqnarray}
Here, we adopt the convention that $R^{-}(H)=\infty$ if $F^{-}(H,k)=\es$. 

\medskip 

\noindent
Notice that 
\begin{equation}
\label{eqn16}
\boxed{R^{+}(H) = \max \; \{\; |V(H)|,\; k\;\}.}
\end{equation}

\medskip

\noindent
Assume that $H$ is the union of two smaller cographs $H_1$ and $H_2$. 
Then 
\begin{equation}
\label{eqn6}
R^{-}(H)  = \min \; \{
R^{-}(H_1)+|V(H_2)|,\;
R^{-}(H_2)+|V(H_1)|, \; 
R^{-}(H_1)+R^{-}(H_2) \;\}. 
\end{equation}
 
\medskip 

\noindent
Now assume that $H$ is the join of two smaller cographs, $H_1$ and $H_2$. 
Then we have 
\begin{equation}
\label{eqn7}
R^{-}(H) = 
\min \;\{\;  
R^{+}(H_1), \;R^{+}(H_2),
\;R^{-}(H_1), \;R^{-}(H_2),\; 2k \; \}.   
\end{equation}

\medskip

\noindent
We prove the correctness of Equation~\eqref{eqn7} below.

\medskip 

\noindent
Let $f$ be a $k$-rainbow function from $F^{-}(H,k)$ 
with minimum cost over $H$.  Consider the following cases.  
\begin{enumerate}[\rm (a)]
\item $f(x) \neq \es$ for all $x \in V(H_1)$. Then there is a vertex 
with an empty label in $H_2$. Let $L_2=\cup_{z \in V(H_2)} f(z)$. 
Define an other $k$-rainbow function $f^{\prime}$ as follows. 
For an arbitrary vertex $x \in V(H_1)$ let $f^{\prime}(x)=f(x) \cup L_2$. 
For all other vertices $y$ in $H_1$ let $f^{\prime}(y)=f(y)$ and for all 
vertices $z$ in $H_2$ let $f^{\prime}(z)=\es$.  Notice that 
$f^{\prime}(V(H_1))=[k]$. So, $f^{\prime}$ is a 
$k$-rainbow function with at most the same cost as $f$. This shows that 
$R^{-}(H)=R^{+}(H_1)$.    
\item $f(y) \neq \es$ for all $y \in V(H_2)$. 
This case is similar to the previous, so that 
$R^{-}(H) = R^{+}(H_2)$. 
\item $f(x) = f(y) = \es$ for some $x \in V(H_1)$ and some $y \in V(H_2)$. 
Let 
\[L_1 = f(V(H_1)) \quad \text{and}\quad 
L_2=f(V(H_2)).\] 
For each color $\ell \in [k]$, 
let $\nu_{\ell}$ be the number of times that $\ell$ is used as a label, 
that is,  
\[\nu_\ell = | \{ \; x \; | \; x \in V(H) \quad 
\text{and}\quad \ell \in f(x)\;\}|.\]   
Consider the following two subcases.
\begin{enumerate}[\rm (i)]
\item 
There exists some $\ell$ with $\nu_\ell = 1$. 
Let $u$ be the unique vertex with 
$\ell \in f(u)$. 
Assume that $u \in V(H_1)$. The case where $u \in V(H_2)$ is 
similar. Then, $u$ is  
adjacent to all $x \in V(H)$ with $f(x) = \es$.  
Modify $f$
to $f'$, such that $f'(u) = f(u) \cup (L_2 \setminus L_1)$, 
$f'(x) = f(x)$ for all $x \in V(H_1)\setminus \{u\}$, 
and $f'(y) = \es$ for all $y \in V(H_2)$.
Then 
$f'$ is a $k$-rainbow function from $F^{-}(H,k)$ and 
the cost of $f'$ is at most the cost of $f$. 
Moreover, $f'$ restricted to $H_1$ 
is a $k$-rainbow function with minimum cost over $H_1$. Thus,
in this case,  $R^{-}(H) = R^{-}(H_1)$.
\item 
For all $\ell$, $\nu_\ell \geq 2$.  
Then, the cost of $f$ over $H$ is at least $2k$.
In this case, we use an alternative function $f'$, 
which selects a certain vertex $u$ 
in $H_1$ and a certain vertex $v$ in $H_2$, 
and set $f'(u) = f'(v) = [k]$. For
all vertices $z \in V(H)\setminus \{u,v\}$, let $f'(z) = \es$.  
The cost of $f'$ is $2k$ (which is at most the cost of $f$), 
and $f'$ remains a $k$-rainbow function from $F^{-}(H,k)$. 
Thus, in this case, $R^{-}(H) = 2k$.
\end{enumerate}
\end{enumerate}
This proves the correctness of Equation~\eqref{eqn7}.

\medskip

\noindent
At the root of the cotree, 
we obtain $\gamma_{rk}(G)$ via  
\begin{equation}
\gamma_{rk}(G) = \min\; \{\; |V(G)|, \; R^{-}(G)\; \}.
\end{equation}

\medskip

\noindent
The cotree can be obtained in linear 
time (see, eg,~\cite{kn:bretscher,kn:corneil2,kn:habib}).  
Each $R^{+}(H)$ is obtained in $O(1)$ time via Equation~\eqref{eqn16}, 
and 
$R^{-}(H)$ is obtained in $O(1)$ time 
via Equations~\eqref{eqn6} and~\eqref{eqn7}. 

\medskip 

\noindent
This proves the theorem. 
\qed\end{proof}

\bigskip 

The weak $\{k\}$-domination number (recall the definition near Equation~\eqref{eqn4}) 
was introduced by Bre\u{s}ar, Henning and 
Rall in~\cite{kn:bresar2} as an accessible, `monochromatic version' of 
$k$-rainbow domination. 
In the following theorem we turn the tables. 

In general, for graphs $G$ one has that $\gamma_{wk}(G) \leq \gamma_{rk}(G)$ 
since, given a $k$-rainbow function $f$ one obtains a 
weak $\{k\}$-dominating function $g$ by defining, for 
$x \in V(G)$,  $g(x)=|f(x)|$. 
The parameters $\gamma_{wk}$ and $\gamma_{rk}$ do not always 
coincide. For example $\gamma_{w2}(C_6)=3$ and $\gamma_{r2}(C_6)=4$. 
Bre\u{s}ar et al. ask, in their Question~3, for which graphs 
the equality $\gamma_{w2}(G)=\gamma_{r2}(G)$ holds. As far as we 
know this problem is still open. Chang et al. 
showed that weak $\{k\}$-domination and $k$-rainbow 
domination are equivalent for strongly chordal graphs~\cite{kn:chang2}. 

\bigskip 

For cographs equality does not hold. For example,   
\begin{equation}
\label{eqn22}
\text{when} \quad  G=(P_3 \oplus P_3) \otimes (P_3 \oplus P_3) 
\quad \text{then}\quad 
\gamma_{w3}(G)=4 \quad \text{and}\quad \gamma_{r3}(G)=6.  
\end{equation}

\bigskip 

Let $G$ be a graph and let $k \in \mathbb{N}$.  
For a function $g: V(G) \rightarrow \{0,\dots,k\}$  
we write 
$\|g\|_{G}=\sum_{x \in V(G)} g(x)$. Furthermore, for $S \subset V(G)$ we write 
$g(S)=\sum_{x \in S} g(x)$. 

\bigskip 
 
\begin{theorem}
\label{thm weak cograph}
There exists an $O(k^2 \cdot n)$ algorithm to compute 
the weak $\{k\}$-domination number for cographs when a cotree is a 
part of the input. 
\end{theorem}
\begin{proof}
Let $k \in \mathbb{N}$. 
For a cograph $H$ and $q \in \mathbb{N} \cup \{0\}$, define 
\begin{multline}
\label{eqn23}
W(H,q)=\min \;  \{\;\|g\|_{H}\;|\; 
g:V(H) \rightarrow \{0,\dots,k\} \quad\text{and}\\
\forall_{x \in V(G)} \; g(x)=0 \quad \Rightarrow \quad g(N(x))+ q \geq k\;\}.
\end{multline}

\medskip

\noindent
When a cograph $H$ is the union of two 
smaller cographs $H_1$ and $H_2$ then 
\begin{equation}
\label{eqn24}
\gamma_{wk}(H)=\gamma_{wk}(H_1)+\gamma_{wk}(H_2).
\end{equation}

\noindent
In such a case, we have
\begin{equation}
\label{eqn24-2}
W(H,q)=  W(H_1,q) + W(H_2,q).  
\end{equation}

\medskip 

\noindent
When a cograph $H$ is the join of two cographs $H_1$ and $H_2$ 
then the minimal cost of a 
weak $\{k\}$-dominating function is bounded from above by $2k$.  Then 
\begin{equation}
\label{eqn25}
W(H,q)= \min \; \{\;W_1+W_2\;|\; 
W_1=W(H_1,q+W_2) \quad\text{and}\quad 
W_2=W(H_2,q+W_1) \;\}.  
\end{equation}

\medskip 

\noindent
The weak $\{k\}$-domination number of a cograph~$G$, $W(G,0)$,  
can be obtained via the above recursion, 
spending $O(k^2)$ time in each of the $n$ nodes in the cotree. 
This completes the proof. 
\qed\end{proof}

\bigskip 

\begin{remark}
A $\{k\}$-dominating function~\cite{kn:domke}  
$g:V(G)\rightarrow \{0,\dots,k\}$  
satisfies 
\[\forall_{x \in V(G)} \;  g(N[x]) \geq k.\]  
The $\{k\}$-domination number $\gamma_{\{k\}}(G)$ 
is the minimal cost of a $\{k\}$-dominating function. 
A similar proof as for Theorem~\ref{thm weak cograph} shows the 
following theorem. 
\begin{theorem}
There exists an $O(k^2 \cdot n)$ algorithm to compute $\gamma_{\{k\}}(G)$ 
when $G$ is a cograph. 
\end{theorem}
Similar results can be obtained for, eg, the $(j,k)$-domination number, 
introduced by Rubalcaba and Slater~\cite{kn:rubalcaba,kn:rubalcaba2}. 
\end{remark}
\begin{remark}
A frequently studied generalization of cographs is the class of 
$P_4$-sparse graphs. 
A graph is $P_4$-sparse if every set of 5 vertices induces at most one  
$P_4$~\cite{kn:hoang,kn:jamison}. 
We show in Appendix~\ref{appendix cographs} that the 
rainbow domination problem can be solved in linear time on 
$P_4$-sparse graphs. 
\end{remark}
  
\section{Weak $\{k\}$-L-domination on trivially perfect graphs}

Chang et al. were able to solve the $k$-rainbow domination problem 
(and the weak $\{k\}$-domination problem) 
for two subclasses of strongly chordal graphs, namely for 
trees and for blockgraphs. 
In order to obtain linear-time algorithms, they   
introduced a variant, called the 
weak $\{k\}$-L-domination problem~\cite{kn:chang2,kn:chang}.   
In this section we show that this problem can be solved in 
$O(k\cdot n)$ time for trivially perfect graphs. 

\begin{definition}
A $\{k\}$-assignment of a graph $G$ is a map $L$ from $V(G)$ to ordered pairs 
of elements from $\{0,\dots,k\}$.  Each vertex $x$ is assigned a 
label $L(x)=(a_x,b_x)$, 
where $a_x$ and $b_x$ are elements of $\{0,\dots,k\}$. 
A weak $\{k\}$-L-dominating function is a function $w:V(G) \rightarrow \{0,\dots,k\}$ 
such that, for each vertex $x$ the following two conditions hold. 
\begin{eqnarray}
w(x) & \geq & a_x, \quad\text{and}\\  
w(x) & =0 &  \quad\Rightarrow\quad w(N[x]) \geq b_x.
\end{eqnarray} 
The weak $\{k\}$-L-domination number is defined as 
\begin{equation}
\label{27}
\gamma_{wkL}(G)=\min\;\{\;\|g\|\;|\; \text{$g$ is a weak 
$\{k\}$-L-dominating function on $G$}\;\}.
\end{equation}
\end{definition}

Notice that 
\begin{equation}
\label{eqn28}
\forall_{x \in V(G)}\; L(x)=(0,k) \quad\Rightarrow \quad 
\gamma_{wk}(G)=\gamma_{wkL}(G).  
\end{equation}

\bigskip 

\begin{definition}
A graph is trivially perfect if it has no induced $P_4$ or $C_4$. 
\end{definition}

Wolk investigated the trivially perfect graphs as the comparability graphs of 
forests. 
Each component of a trivially perfect graph $G$ has a model which is a rooted tree $T$  
with vertex set $V(G)$.  Two vertices of $G$ are adjacent if, in $T$, one lies on the path 
to the root of the other one. Thus each path from a leaf to the root is a maximal clique 
in $G$ and these are all the maximal cliques.    
See~\cite{kn:chu,kn:golumbic} for the recognition of these graphs. 
In the following we assume that a rooted tree $T$ as a model for the graph is a part of 
the (connected) input. 

\bigskip 

We simplify the problem by using two basic observations. 
(See~\cite{kn:chang2,kn:chang} for similar observations.)    
Let $T$ be a rooted tree which is the model for a connected trivially perfect graph $G$. 
Let $R$ be the root of $T$; note that this is a universal vertex in $G$. 
We assume that $G$ is equipped with a $\{k\}$-assignment $L$, which attributes 
each vertex $x$ with a pair 
$(a_x,b_x)$ of numbers from $\{0,\dots,k\}$. 
\begin{enumerate}[\rm (I)]
\item There exists a weak $\{k\}$-L-dominating function $g$ of 
minimal cost such that 
\begin{equation}
\label{eqn29}
\forall_{x \in V(G)\setminus \{R\}} \; a_x > 0 \quad \Rightarrow \quad g(x)=a_x.
\end{equation}
\item There exists a weak $\{k\}$-L-dominating function $g$ of minimal cost 
such that 
\begin{equation}
\label{eqn32}
\forall_{x \in V(G) \setminus \{R\}} \; 
a_x=0 \quad\text{and}\quad b_x \leq \sum_{y \in N[x]} a_y \quad \Rightarrow \quad g(x)=0.
\end{equation}
\end{enumerate}

\bigskip 

\begin{definition}
The \underline{reduced instance} of the weak $\{k\}$-L-domination problem 
is the subtree $T^{\prime}$ of $T$ with vertex set $V(G^{\prime}) \setminus W$, where 
\begin{multline}
\label{eqn30}
W = \{\; x \;|\; x \in V(G) \setminus \{R\} \quad 
\text{and}\quad a_x >0\;\} \quad \cup \quad  \\ 
\{\;x \;|\; x \in V(G) \setminus \{R\} \quad\text{and}\quad a_x =0 
\quad\text{and}\quad \sum_{y \in N[x]} a_y \geq b_x\;\}.   
\end{multline}
The labels of the reduced instance are, 
for $x \neq R$, $L(x)=(a_x^{\prime},b_x^{\prime})$, where  
\begin{equation}
\label{eqn31} 
a_x^{\prime}=0 \quad \text{and}\quad 
b_x^{\prime}=b_x - \sum_{y \in N[x]} a_y, 
\end{equation}
and the root $R$ has a label $L(R)=(a_R^{\prime},b_R^{\prime})$, where    
\begin{equation}
\label{eqn33} 
a_R^{\prime}=a_R \quad\text{and}\quad   
b_R^{\prime}=\max\;\{\;0,\;b-\sum_{x \in V(G) \setminus \{R\}} a_x\;\}.
\end{equation}
\end{definition}

\bigskip 

The previous observations prove the following lemma. 

\begin{lemma}
Let $T^{\prime}$ and $L^{\prime}$ be a reduced instance of a weak 
$\{k\}$-L-domination problem. 
Then 
\[\gamma_{wkL}(G)=\gamma_{wkL^{\prime}}(G^{\prime})+
\sum_{x \in V(G)\setminus \{R\}} a_x.\] 
\end{lemma}
  
\bigskip 

In the following, let $G$ be a connected, trivially perfect graph and 
let $G$ be equipped with a $\{k\}$-assignment. 
Let $G^{\prime}=(V^{\prime},E^{\prime})$ be a reduced instance with model 
a $T^{\prime}$ and a root $R$, and a reduced assignment $L^{\prime}$.  
Let $g$ be a weak $\{k\}$-$L^{\prime}$-dominating function on $G^{\prime}$ of 
minimal cost. 
Notice that we may assume that 
\[\boxed{\forall_{x \in V(G^{\prime}) \setminus \{R\}} \; g(x) \in \{0,1\}.}\] 

\bigskip 

Let $x$ be an internal vertex in the tree $T^{\prime}$ 
and let $Z$ be the set of descendants of $x$. Let $P$ be the path in 
$T^{\prime}$ from $x$ to the root $R$.   
Assume that $Z$ is a union of  
$r$ distinct cliques, say $B_1,\dots,B_r$. 
Assume that the vertices of each $B_j$ are ordered 
$x^j_1,\dots,x^j_{r_j}$ such that 
\[\boxed{p \leq q \quad\Rightarrow\quad b^{\prime}_{x^j_p} \geq b^{\prime}_{x^j_q}.}\] 
Define $d_{x^j_p}=b^{\prime}_{x^j_p}-p+1$.  
Relabel the vertices of $Z$ as $z_1,\dots,z_{\ell}$ 
such that 
\[\boxed{p \leq q \quad\Rightarrow d_{z_p} \geq d_{z_q}.}\] 

\bigskip 

\begin{lemma}
There exists an optimal weak $\{k\}$-$L^{\prime}$-dominating function 
$g$ such that $g(z_i) \geq g(z_j)$ when $i < j$. 
\end{lemma}
We moved the proof of this lemma to Appendix~\ref{appendix TP}. 

\bigskip 

\begin{definition}
For $a \in \{0,\dots,k\}$, $a \geq a_R^{\prime}$, 
let $\Gamma(G^{\prime},L^{\prime},a)$ 
be the minimal cost over all weak $\{k\}$-$L^{\prime}$-dominating 
functions $g$ on $G^{\prime}$ on condition that $g(P) \geq a$. 
\end{definition}

\bigskip 

\begin{lemma}
Define $d_{z_{\ell+1}}=a$. Let $i^{\ast} \in \{1,\dots,\ell+1\}$ be such that 
\begin{enumerate}[\rm (a)]
\item $\max\;\{\;a,\;d_{z_i^{\ast}}\;\}+i^{\ast}-1$ is smallest possible, and 
\item $i^{\ast}$ is smallest possible with respect to $(a)$. 
\end{enumerate}
Let $H=G^{\prime}-Z$. 
Let $L^H$ be the restriction of $L^{\prime}$ to $V(H)$ with the following 
modifications. 
\[\forall_{y \in P}\; b^H_y=\max\;\{\;0,\;b^{\prime}_y-i^{\ast}+1\;\}.\] 
Let $a^H=\max\;\{\;a,\;d_{z_{i^{\ast}}}\;\}$. 
Then 
\[\Gamma(G^{\prime},L^{\prime},a)=\Gamma(H,L^H,a^H)+i^{\ast}-1.\] 
\end{lemma}
We moved the proof of this lemma to Appendix~\ref{appendix TP}. 
 
\bigskip 

The previous lemmas prove the following theorem. 

\begin{theorem}
Let $G$ be a trivially perfect graph with $n$ vertices.  
Let $T$ be a rooted tree that represents $G$. 
Let $k \in \mathbb{N}$ and let $L$ be a $\{k\}$-assignment of $G$. 
Then there exists an $O(k \cdot n)$ algorithm that computes a 
weak $\{k\}$-L-dominating function of $G$. 
\end{theorem}

\bigskip 

The related $(j,k)$-domination problem can be solved in linear time on 
trivially perfect graphs. 
The weak $\{k\}$-L-domination problem can be solved in 
linear time on complete bipartite graphs. We moved that section 
to Appendix~\ref{appendix CB}. 
 
\section{$2$-Rainbow domination of interval graphs}
\label{section interval}

In~\cite{kn:bresar2} 
the authors ask four questions, the last one of 
which is, whether there is a polynomial  
algorithm for the $2$-rainbow domination problem on (proper) interval graphs. 
In this section we show that $2$-rainbow domination can be solved 
in polynomial time on interval graphs. 

\bigskip 

We use the equivalence of the $2$-rainbow domination problem 
with the 
weak $\{2\}$-domination problem. The equivalence of the two problems, 
when restricted to trees and interval graphs, was observed in~\cite{kn:bresar2}. 
Chang et al., proved that it holds for general $k$ when restricted to the 
class of strongly 
chordal graphs~\cite{kn:chang2}. The class of interval graphs is properly 
contained in that of the strongly chordal graphs. 

\bigskip 

An interval graph has a consecutive clique arrangement. That is a 
linear ordering $[C_1,\dots,C_t]$ of the maximal cliques of the interval graph 
such that, for each vertex, the cliques that contain it occur consecutively 
in the ordering~\cite{kn:gilmore}. 

\bigskip 

For a function $g: V(G) \rightarrow \{0,1,2\}$ we write, as usual, 
for any $S \subseteq V(G)$, 
\[g(S)=\sum_{x \in S} g(x).\] 
The weak $\{2\}$-domination problem is defined as follows. 

\begin{definition}
Let $G$ be a graph. A function $g: V(G) \rightarrow \{0,1,2\}$ 
is a weak $\{2\}$-dominating function 
on $G$ if 
\begin{equation}
\label{eqn19}
\forall_{x \in V(G)} \;\; g(x)=0 \quad \text{implies}\quad g(N[x]) \geq 2.
\end{equation}
The weak $\{2\}$-domination number of $G$ is 
\begin{equation}
\label{eqn20}
\gamma_{w2}(G)=\min \;\{\;\sum_{x \in V(G)}\; g(x) \;|\; 
\text{$g$ is a weak $\{2\}$-domination function on $G$}\;\}.
\end{equation}
\end{definition}

Br\u{s}ar and \u{S}umenjak proved the following theorem~\cite{kn:bresar2}. 

\begin{theorem}
When $G$ is an interval graph, 
\begin{equation}
\label{eqn21}
\gamma_{w2}(G)=\gamma_{r2}(G).
\end{equation}
\end{theorem}

\bigskip 

In the following, let $G=(V,E)$ be an interval graph. 

\begin{lemma}
\label{lm int1}
There exists a weak $\{2\}$-dominating function $g$, with $g(V)=\gamma_{r2}(G)$, 
such that every maximal clique has at most 2 vertices assigned the value 2. 
\end{lemma}
\begin{proof}
Assume that $C_i$ is a maximal clique in the consecutive clique arrangement of $G$. 
Assume that $C_i$ has 3 vertices $x$, $y$ and $z$ with $g(x)=g(y)=g(z)=2$. 
Assume that, among the three of them, $x$ has the most neighbors 
in $\cup_{j\geq i} C_j$ and that $y$ has the most neighbors in $\cup_{j \leq i} C_j$. 
Then any neighbor of $z$ is also a neighbor of $x$ or it is a neighbor of $y$. 
So, if we redefine $g(z)=1$, we obtain a weak $\{2\}$-dominating function 
with value less than $g(V)$, a contradiction. 
\qed\end{proof}

\bigskip 

\begin{lemma}
\label{lm int2}
There exists a weak $\{2\}$-dominating function $g$ with minimum value 
$g(V)=\gamma_{r2}(G)$ such that every maximal clique 
has at most four vertices with value 1. 
\end{lemma}
\begin{proof}
The proof is similar to that of Lemma~\ref{lm int1}. 
Let $C_i$ be a clique in the consecutive clique arrangement of $G$. 
Assume that $C_i$ has 5 vertices $x_i$, $i \in \{1,\dots,5\}$, with 
$g(x_i)=1$ for each $i$. Order the vertices $x_i$ according to their 
neighborhoods in $\cup_{j \geq i} C_j$ and according to their 
neighborhoods in $\cup_{j \leq i} C_j$.  For simplicity, assume that 
$x_1$ and $x_2$ have the most neighborhoods in the first union of cliques 
and that $x_3$ and $x_4$ have the most neighbors in the second union of 
cliques. Then $g(x_5)$ can be reduced to zero; any other vertex that has $x_5$ 
in its neighborhood already has two other $1$'s in it. 

\medskip 

\noindent  
This proves the lemma. 
\qed\end{proof}

\bigskip 

\begin{theorem}
\label{thm int}
There exists a polynomial algorithm to compute the 
$2$-rainbow domination number for interval graphs. 
\end{theorem}
We moved the proof of this theorem and some remarks 
to Appendix~\ref{appendix interval}. 

We obtained similar results for the class of permutation graphs. 
We moved that section to Appendix~\ref{appendix permutation}. 

\section{$\NP$-Completeness for splitgraphs}

A graph $G$ is a splitgraph if $G$ and $\Bar{G}$ are both chordal. 
A splitgraph has 
a partition of its vertices into two set $C$ and $I$, 
such that the subgraph induced by $C$ is a clique and 
the subgraph induced by $I$ is an independent set.
 
\bigskip

Although the $\NP$-completeness of $k$-rainbow 
domination for chordal graphs was established in~\cite{kn:chang}, their proof 
does not imply the intractability for the class of splitgraphs. 
However, that is easy to mend. 

\bigskip 

\begin{theorem}
\label{thm NP-c rainbow}
Let $k \in \mathbb{N}$. Computing $\gamma_{rk}(G)$ 
is $\NP$-complete for splitgraphs. 
\end{theorem}
We moved the proof of this theorem to Appendix~\ref{appendix splitgraphs}. 

\bigskip 

Similarly, we have the following theorem. 

\begin{theorem}
\label{thm NP-c weak}
Let $k \in \mathbb{N}$. Computing $\gamma_{wk}(G)$ 
is $\NP$-complete for splitgraphs. 
\end{theorem}
We moved the proof of this theorem to Appendix~\ref{appendix splitgraphs}.

\newpage 

\appendix 

\section{Rainbow domination on $P_4$-sparse graphs}
\label{appendix cographs}

Many problems can be solved in linear time for the class of 
cographs and $P_4$-sparse graphs, 
see eg,~\cite[Theorem~2 and Corollary~3]{kn:courcelle2}. 
Both classes are of bounded cliquewidth (or rankwidth). 
We are not aware of many problems 
of which the complexities differ on the two classes of graphs.  
An interesting example, that might have different complexities on 
cographs and $P_4$-sparse graphs, 
is the weighted coloring problem~\cite{kn:araujo}. 
In this section we show that the $k$-rainbow domination problem 
can be solved in linear time on $P_4$-sparse graphs. 

\bigskip 

Ho\`ang introduced $P_4$-sparse graphs as follows. 

\begin{definition}
A graph is $P_4$-sparse if every set of 5 vertices contains at most 
one $P_4$ as an induced subgraph. 
\end{definition}

\bigskip 

Jamison and Olariu showed that $P_4$-sparse graphs have a decomposition 
similar to the decomposition of cographs. To define it we need 
the concept of a spider. 

\begin{definition}
A graph $G$ is a thin spider if its vertices can be partitioned into three 
sets $S$, $K$ and $T$, such that 
\begin{enumerate}[\rm (a)]
\item $S$ induces an independent set and $K$ induces a clique and 
\[|S|=|K| \geq 2.\] 
\item Every vertex of $T$ is adjacent to every vertex of $K$ and to no 
vertex of $S$. 
\item There is a bijection from $S$ to $K$, such that every 
vertex of $S$ is adjacent to exactly one unique vertex of $K$. 
\end{enumerate}
A thick spider is the complement of a thin spider. 
\end{definition}
Notice that, possibly, $T=\es$. The set $T$ is called the head, 
and $K$ and $S$ are called the body and the feet of the thin spider. 

\bigskip 

Jamison and Olariu proved the following theorem~\cite{kn:jamison}. 

\begin{theorem}
A graph is $P_4$-sparse if and only if for every induced subgraph $H$ 
one of the following holds. 
\begin{enumerate}[\rm (i)]
\item $H$ is disconnected, or 
\item $\Bar{H}$ is disconnected, or 
\item $H$ is isomorphic to a spider. 
\end{enumerate}
\end{theorem}
 
\bigskip 

\begin{theorem}
\label{P4sparse}
There exists a linear-time algorithm to compute the 
$k$-rainbow domination number for $P_4$-sparse graphs and $k \in \mathbb{N}$. 
\end{theorem}
\begin{proof}
We extend Formula~\eqref{eqn7} in the proof of 
Theorem~\ref{thm cograph} to include spiders. 

\medskip 

\noindent 
Assume that $G$ is a thin spider, with a head $T$, a body $K$ and 
an independent set of feet $S$.  
We need to consider only $k$-rainbow colorings such that at least one 
vertex of $G$ has an empty label.  
Notice that we may assume that all 
the feet have labels of cardinality at most one. Furthermore, 
there is at most one vertex in $S$ that has an empty label. 

\medskip 

\noindent 
First assume that one foot, say $x$, has an empty label. Then its 
neighbor, say $y \in K$ has label $[k]$. In that case, the (only) optimal 
$k$-rainbow coloring has 
\begin{enumerate}[\rm (a)]
\item all vertices in $S \setminus \{x\}$ a label of cardinality 1, 
\item all vertices in $K \setminus \{y\}$ an empty label, and 
\item all vertices of $T$ also an empty label. 
\end{enumerate}
It follows that the cost in this case is 
\begin{equation}
\label{eqsparse1}
|S|-1+k. 
\end{equation}

\medskip 

\noindent
All other optimal $k$-rainbow colorings give all feet a label 
of cardinality 1. 
If some vertex of $T$ has an empty label, then its neighborhood has cost 
at least $k$, and the total cost will be more than~\eqref{eqsparse1}.  
So, the only alternative $k$-rainbow coloring assigns no empty label to 
$S$ nor $T$, and assigns an empty label to some vertex $v$ of $K$.  In such 
a case, the neighborhood of $v$ has cost at least $k$,
so that combining with the cost in the non-neighboring feet,
the total cost is at least that of~\eqref{eqsparse1}.
   
\medskip 

\noindent
In conclusion, the optimal cost of a thin spider $G$ is 
\begin{equation}
\min\; \{\; |V(G)|, |S| - 1 + k\; \}.
\end{equation}

\medskip 

\noindent
The case analysis for the thick spider is similar. 
This proves the theorem. 
\qed\end{proof}

\bigskip 

\begin{remark} 
A graph $G$ 
is $P_4$-lite if every induced subgraph $H$ of at most 6 vertices satisfies 
one of the following. 
\begin{enumerate}[\rm (i)]
\item $H$ contains at most two induced $P_4$'s, or 
\item $H$ is isomorphic to $S_3$ or $\Bar{S_3}$. 
\end{enumerate}
These graphs contain the $P_4$-sparse graphs. 

\medskip 

Babel and Olariu studied $(q,q-4)$-graphs, see, eg,~\cite{kn:babel}. 
It would be interesting to study the 
complexity of the $k$-rainbow problem on these graphs. 
\end{remark}
     
\section{Trivially perfect graphs}
\label{appendix TP}

\begin{lemma}
There exists an optimal weak $\{k\}$-$L^{\prime}$-dominating function
$g$ such that $g(z_i) \geq g(z_j)$ when $i < j$.
\end{lemma}
\begin{proof}
Assume not.
Let $\Hat{g}$ be the same as $g$ except that $\Hat{g}(z_i)=1$ and
$\Hat{g}(z_j)=0$.
To see that $\Hat{g}$ is a weak $\{k\}$-L-dominating function,
first notice that, when $z_i$ and $z_j$ are in a common clique of
$G^{\prime}[Z]$
then it is easy to see that $\Hat{g}$
defines a weak $\{k\}$-L-dominating function
of at most the same cost.

\medskip

\noindent
Now assume that the claim holds for all
pairs contained in common cliques.
Then choose $z_i$ and $z_j$ with
$i < j$, $g(z_i)=0$ and $g(z_j)=1$, and $z_i$ and $z_j$
in different cliques with $i$ as small as possible
and $j$ as large as possible. Say
$z_i=x^m_{m^{\ast}}$ and $z_j=x^h_{h^{\ast}}$. Then, by assumption,
\[\forall_{t < m^{\ast}}\; g(x^m_t)=1 \quad 
\forall_{t > m^{\ast}}\; g(x^m_t)=0 \quad  
\forall_{t < h^{\ast}}\; g(x^h_t)=1 \quad 
\forall_{t > h^{\ast}}\; g(x^h_t)=0.\]
Now, notice that
\[g(N(z_i))=m^{\ast}-1+g(P) \geq b^{\prime}_{z_i} \quad\Rightarrow\quad 
 g(P) \geq b^{\prime}_{z_i}-m^{\ast}+1.\]
We prove that $\Hat{g}(N[z_j]) \geq b^{\prime}_{z_j}$.
We have that
\[\Hat{g}(N[z_j])= g(P) +h^{\ast} -1 \geq b^{\prime}_{z_i}-m^{\ast}+h^{\ast}.\]
By definition and the ordering on $Z$,
\[b^{\prime}_{z_i}-m^{\ast}+1 = d_{z_i} \geq d_{z_j}=b^{\prime}_{z_j}-h^{\ast}+1 
\quad\Rightarrow \quad 
\Hat{g}(N[z_j]) \geq b^{\prime}_{z_j}.\]
This proves the lemma.
\qed\end{proof}

\bigskip 

\begin{lemma}
\label{lm appendix TP}
Define $d_{z_{\ell+1}}=a$. Let $i^{\ast} \in \{1,\dots,\ell+1\}$ be such that 
\begin{enumerate}[\rm (a)]
\item $\max\;\{\;a,\;d_{z_i^{\ast}}\;\}+i^{\ast}-1$ is smallest possible, and 
\item $i^{\ast}$ is smallest possible with respect to $(a)$. 
\end{enumerate}
Let $H=G^{\prime}-Z$. 
Let $L^H$ be the restriction of $L^{\prime}$ to $V(H)$ with the following 
modifications. 
\[\forall_{y \in P}\; b^H_y=\max\;\{\;0,\;b^{\prime}_y-i^{\ast}+1\;\}.\] 
Let $a^H=\max\;\{\;a,\;d_{z_{i^{\ast}}}\;\}$. 
Then 
\[\Gamma(G^{\prime},L^{\prime},a)=\Gamma(H,L^H,a^H)+i^{\ast}-1.\] 
\end{lemma}
\begin{proof}
Let $g$ be a weak $\{k\}$-$L^{\prime}$-dominating 
function, satisfying $g(P) \geq a$, of minimal cost. 
We first prove that the restriction of $g$ to $H$ is a 
weak $\{k\}$-$L^H$-dominating 
function and that $g(P) \geq a^H$. 

\medskip 

\noindent
Let $i$ be the largest index such that $g(z_j)=1$ for all $j < i$. 
If $i=\ell+1$ we have $g(P) \geq a=d_{z_{\ell+1}}$. 

\medskip 

\noindent
Now assume that $i \leq \ell$. 
Let $z_i=x^m_{m^{\ast}}$.   
Then 
\[g(N[z_i]) = g(P)+m^{\ast}-1 \geq b^{\prime}_{z_i} \quad\Rightarrow\quad 
g(P) \geq b^{\prime}_{z_i} - m^{\ast}+1=d_{z_i}.\] 
Thus $g(P) \geq \max \{\;a,\;d_{z_i}\;\}$ and 
so we have that   
\[g(P) +i-1 \geq \max \;\{\;a,\;d_{z_i}\;\} +i-1 \geq 
\max\;\{\;a,\;d_{z_{i^{\ast}}}\;\}+i^{\ast}-1.\] 
We claim that $i^{\ast} \geq i$. 

\medskip 

\noindent
Suppose that $i^{\ast} < i$. Then let $\Hat{g}$ be the same as $g$ 
except that 
\[\Hat{g}(R)= \min\;\{\;g(R)+i - i^{\ast},\; k\;\} 
\quad\text{and}\quad 
\forall_{i^{\ast} \leq j < i} \; \hat{g}(z_j)=0.\] 
Let $j \geq i^{\ast}$ and let $h$ and $h^{\ast}$ be such that $z_j=x^h_{h^{\ast}}$. 
Let $t$ be the smallest index for which $\hat{g}(x^h_t)=0$. 
By the inequality above, $\Hat{g}(P)\geq \max\{\;a,\;d_{z_i^{\ast}}\;\}$, and so,  
\[\hat{g}(N[x^h_{h^{\ast}}]) = \hat{g}(N[x^h_t])=\hat{g}(P)+t-1 \geq 
d_{z_{i^{\ast}}}+t-1 \geq d_{x^h_t}+t-1 =b^{\prime}_{x^h_t} \geq b^{\prime}_{x^h_{h^{\ast}}}.\] 
Thus $\hat{g}$ is a weak $\{k\}$-$L^{\prime}$-dominating function on $G^{\prime}$ with 
$\Hat{g}(P) \geq a$ and with cost at most $\|g\|$. 
Therefore, we may assume that 
\[i^{\ast} \geq i \quad\text{and}\quad  
g(P) \geq  
\max \; \{\;a,\;d_{z_i^{\ast}}\;\} = a^H.\]  

\medskip 

\noindent
Also, notice that   
\begin{multline}
\forall_{y \in P} \; g(y)=0 \quad\Rightarrow \\
g^H(N[y]) \geq \max \{\;0,\; g(N[y])-i+1\;\} \geq 
\max\;\{\;0,\;b_y^{\prime}-i^{\ast}+1\;\}=b_y^H.
\end{multline}
This proves that 
\[\Gamma(G^{\prime},L^{\prime},a) - i^{\ast}+1 \geq \Gamma(H,L^H,a^H).\] 
 
\medskip 

\noindent 
Now let $g^H$ be a weak $\{k\}$-$L^H$-dominating function of $H$,  
with $g^H(P) \geq a^H$, of minimal cost. 
Extend $g^H$ to a function $g^{\prime}$ on $G^{\prime}$ by defining 
$g^{\prime}(z_j)=1$ for all $j < i^{\ast}$ and $g^{\prime}(z_j)=0$ for all 
$j \geq  i^{\ast}$. We claim that $g^{\prime}$ is a weak $\{k\}$-L-dominating 
function of $G^{\prime}$. 
Let $i \geq i^{\ast}$. 
Say $z_i=x^m_{m^{\ast}}$.  Let $m^{\ast}$ be the first index 
such that 
$g^{\prime}(x^m_j)=0$ for $j \geq m^{\ast}$. 
Then   
\[g^{\prime}(N[z_i])=g^{\prime}(P)+m^{\ast}-1 \geq d_{z_i^{\ast}}+m^{\ast}-1 
\geq d_{z_i}+m^{\ast}-1 = b_{z_i}.\]    
For $i < i^{\ast}$ we have $g^{\prime}(z_i) \neq 0$. For the vertices $y \in P$, 
$g^{\prime}(N[y])=k$ or 
\[g^{\prime}(N[y]) \geq g^H(N_H[y])+i^{\ast}-1 \geq b^H_y+i^{\ast}-1 \geq b_y^{\prime}.\] 
This proves the lemma. 
\qed\end{proof}

\section{$\NP$-Completeness proofs for splitgraphs}
\label{appendix splitgraphs}
 
\begin{theorem}
\label{appendix thm NP-c rainbow}
For each $k \in \mathbb{N}$, the $k$-rainbow domination problem 
is $\NP$-complete for splitgraphs. 
\end{theorem}
\begin{proof}
Assume that $G$ is a splitgraph with maximal clique $C$ and independent set $I$. 
Construct an auxiliary graph $G^{\prime}$ by making $k-1$ pendant 
vertices adjacent to each vertex of $C$. 
Thus $G^{\prime}$ has $|V(G)|+|C|(k-1)$ 
vertices, and $G^{\prime}$ remains a splitgraph. 
We prove that
\[\gamma_{rk}(G^{\prime})=\gamma(G) + |C|\cdot (k-1).\] 
Since 
domination is $\NP$-complete for splitgraphs~\cite{kn:bertossi}, this 
proves that $k$-rainbow domination is $\NP$-complete also. 

\medskip

\noindent
We first show that 
\[\gamma_{rk}(G^{\prime}) \leq \gamma(G) + |C|\cdot (k-1).\] 
Consider a dominating set $D$ of $G$ with 
$|D|=\gamma(G)$.  
We use $D$ to construct a $k$-rainbow function $f$ for $G^{\prime}$ as follows:
\begin{itemize}
\item For any $v \in D$, if $v \in C$, let $f(v) = [k]$;  else, if $v \in I$, 
let $f(v) = \{k\}$;
\item For any $v \in V(G)\setminus D$, let $f(v) =  \es$;
\item For the $k-1$ pendant vertices attaching to a vertex $v \in C$, 
if $f(v) = [k]$, 
then $f$ assigns to each of these pendant vertices an empty set.  
Otherwise, if $f(v) = \es$, then $f$ assigns the distinct size-1 sets 
$\{1\}, \{2\}, \ldots, \{k-1\}$ to these pendant vertices, respectively.
\end{itemize}
It is straightforward to check that $f$ is a $k$-rainbow function. 
Moreover, we have
\begin{equation}
\gamma_{rk}(G^{\prime}) \leq \sum_{x \in V(G^{\prime})} \: |f(x)| 
                                              = \gamma(G) + |C|\cdot (k-1).
\end{equation}

\medskip

\noindent
We now show that 
\[\gamma_{rk}(G^{\prime}) \geq \gamma(G) + |C|\cdot (k-1).\] 
Consider a minimizing $k$-rainbow function $f$ for $G^{\prime}$. 
Without loss of generality, 
we further assume that $f$ assigns either $\es$ or a size-1 subset 
to each pendant vertex.%
\footnote{Otherwise, if a pendant vertex $p$ attaching $v$ 
is assigned a set with two or more labels, say 
$f(p) = \{\ell_1, \ell_2, \ldots\}$, we 
modify $f$ into $f'$ so that 
$f'(p) = \{ \ell_1 \}$, $f'(v) = f(v) \cup (f(p)\setminus \{\ell_1\})$, 
and $f'(x) = f(x)$ for the remaining vertices; 
the resulting $f'$ is still a minimizing $k$-rainbow function.}
Define $D \subseteq V(G)$ as  
\begin{equation}
D = \{\; x \; \mid\; f(x) \neq \es \;\mbox{\ and\ }\; x \in V(G)\; \}.
\end{equation}
That is, 
$D$ is formed by removing all the pendant vertices in $G^{\prime}$, 
and selecting all those vertices where $f$ assigns a non-empty set.  
Observe that $D$ is a dominating set of $G$.%
\footnote{That is so because for any $v \in V(G)\setminus D$, 
we have $f(v) = \es$ so that the union of labels of $v$'s 
neighbor in $G^{\prime}$ is $[k]$;  
however, at most $k-1$ neighbors of $v$ are removed, 
and each was assigned a size-1 set, so that $v$ must have at 
least one neighbor in $D$.}  
Moreover, we have

\begin{eqnarray*}
|D| &=& \sum_{x \in C}\; [f(x) \neq \es]\;  +\; \sum_{x \in I}\; [f(x) \neq \es]  \\
&\leq& \sum_{x \in V(G^{\prime}) \setminus I}\; |f(x)|  - |C|\cdot (k-1) + 
\sum_{x \in I}\; |f(x)|  \\
&\leq& \sum_{x \in V(G^{\prime})} |f(x)| -  |C|\cdot (k-1), 
\end{eqnarray*}
where the first inequality follows from the fact that for 
each $v \in C$ and its corresponding pendant vertices $P_v$, 
\[|f(v)| + \sum_{x \in P_v} |f(x)| - (k-1) = 
\begin{cases}
0 &  
\text{if $f(v) = \es$}\\
 \geq 1 & \text{if $f(v) \neq \es$.}
\end{cases}\]  
Consequently, we have
\begin{equation}
\gamma(G) \leq |D| \leq \gamma_{rk}(G^{\prime}) - |C|\cdot (k-1). 
\end{equation}
This proves the theorem. 
\qed\end{proof}

\bigskip 
\begin{theorem}
\label{appendix thm NP-c weak}
For each $k \in \mathbb{N}$, the weak $\{k\}$-domination 
problem is $\NP$-complete for splitgraphs. 
\end{theorem}
\begin{proof}
Let $G$ be a splitgraph with maximal clique $C$ and independent set $I$. 
Construct the graph $G^{\prime}$  as in 
Theorem~\ref{thm NP-c rainbow}, by adding $k-1$ pendant vertices to each 
vertex of the maximal clique $C$. 
We prove that 
\[\gamma_{wk}(G^{\prime})=\gamma(G)+|C|\cdot (k-1).\]

\medskip 

\noindent
First, let us prove that 
\[\gamma_{wk}(G^{\prime}) \leq \gamma(G)+|C|(k-1).\]
Let $D$ be a minimum dominating set. 
Construct a weak $\{k\}$-domination function 
$g: V(G^{\prime}) \rightarrow \{0,\dots,k\}$ as follows. 
\begin{enumerate}[\rm (i)]
\item For $x \in D \cap C$, let $g(x)=k$. 
\item For $x \in D \cap I$, let $g(x)=1$. 
\item For $x \in V(G) \setminus D$, let $g(x)=0$. 
\item For a pendant vertex $x$ with $N(x) \in D$, let $g(x)=0$.
\item For a pendant vertex $x$ with $N(x) \notin D$, let $g(x)=1$. 
\end{enumerate}
It is easy to check that $g$ is a weak $\{k\}$-dominating function 
with cost 
\[\gamma_{wk}(G^{\prime}) \leq \sum_{x \in V(G^{\prime})} g(x) = 
\gamma(G)+|C| \cdot (k-1).\] 

\medskip 

\noindent 
To prove the converse, let $g$ be a weak $\{k\}$-dominating function 
for $G^{\prime}$ of minimal cost. We may assume that $g(x) \in \{0,1\}$ for 
every pendant vertex $x$. 
Define 
\[D=\{\;x\;|\; x \in V(G) \quad\text{and}\quad g(x) > 0\;\}.\] 
Then $D$ is a dominating set of $G$. 
Furthermore, 
\begin{eqnarray*}
\gamma(G) \leq |D| &=& \sum_{x \in C}\; [g(x) > 0] + \sum_{x \in I}\; [g(x)>0]\\
& \leq & \sum_{x \in V(G^{\prime}) \setminus I} g(x) - |C|\cdot(k-1)+\sum_{x \in I} g(x) \\
& \leq & \sum_{x \in V(G^{\prime})}g(x) - |C| \cdot (k-1) \\
& \leq & \gamma_{wk}(G^{\prime}) - |C| \cdot (k-1). 
\end{eqnarray*}
This proves the theorem.
\qed\end{proof}

\section{2-Rainbow domination of interval graphs}
\label{appendix interval}

\begin{theorem}
\label{appendix thm int}
There exists a polynomial algorithm to compute the
$2$-rainbow domination number for interval graphs.
\end{theorem}
\begin{proof}
By~Lemmas~\ref{lm int1} and~\ref{lm int2} there is a polynomial
dynamic programming algorithm which solves the problem. Let
$[C_1,\dots,C_t]$ be a consecutive clique arrangement. For each $i$
the algorithm computes a table of partial $2$-rainbow domination numbers
for the subgraph induced by $\cup_{\ell =1}^i C_{\ell}$, parameterized
by a given subset of vertices in $C_i$ that are assigned the values 0, 1 and 2.
We say that a vertex $x \in C_i$ is {\em satiated\/} if
\[g(N(x) \cap (\cup_{\ell=1}^i C_i)) \geq 2.\]
The tabulated rainbow domination numbers are partial in the sense that
the neighborhood condition is not necessarily satisfied for all vertices
in $C_i$ that are assigned the value 0. The vertices of $C_i$
that are assigned the value 0
which are not satiated, either need
an extra 1 or an extra 2. Each of the two subsets of nonsatiated
vertices is characterized
by one representative vertex, the one among them that extends
furthest to the left.

\medskip

\noindent
In total, the dynamic system is characterized by 4 state variables.
They are
\begin{enumerate}[\rm (i)]
\item the set of, at most two 2's,
\item the set of, at most four 1's,
\item the nonsatiated vertex that
extends furthest to the left and that needs an extra 1, and
\item a similar nonsatiated
vertex that needs an extra 2.
\end{enumerate}
Each clique has at most $n=|V(G)|$ vertices and so, there are
at most $n^2 \cdot n^4 \cdot n \cdot n=n^8$ different assignments of the
state variables.
In the transition $i \rightarrow i+1$, the table is
computed by minimizing, for all sensible assignments of vertices in $C_{i+1}$,
over the compatible values in the table at stage $i$.
Each update takes $O(1)$ time.
The number of cliques is at most $n$, and so the
algorithm runs in $O(n^9)$ time.
\qed\end{proof}

\bigskip

\begin{remark}
The observations above can be generalized to show that,
for each $k$, the $k$-rainbow domination problem can be solved
in polynomial time for interval graphs. However, this leaves open
the question whether $k$-rainbow domination is in $\FPT$ for interval graphs.
\end{remark}

\bigskip

\begin{remark}
Let $A$ be the closed neighborhood matrix of an interval
graph or a  strongly chordal graph $G$. Then $A$ is totally balanced, that is,
after a suitable permutation or rows and columns the matrix becomes greedy, ie,
$\Gamma$-free. A matrix is greedy if it has no $2 \times 2$ submatrix
$\bigl ( \begin{smallmatrix} 1 & 1 \\ 1 & 0 \end{smallmatrix} \bigr )$. 
For (totally) balanced matrices some polyhedra have only integer extreme points. 
This leads to polynomial algorithms for domination on eg, strongly chordal graphs.  
Notice that the closed neighborhood matrix of $G \Box K_2$ is
\begin{equation}
\label{eqn11}
B=\begin{pmatrix} A & I \\ I & A \end{pmatrix}.
\end{equation}
To solve the $2$-rainbow domination problem for strongly chordal
graphs, one needs to solve the following integer programming problem.
\begin{eqnarray}
\label{eqn12}
\min \;  \sum_{i=1}^{2n} \; x_i \quad & & \text{such that}\\ 
 && B \mathbf{x} \geq \mathbf 1 \quad \text{and}\quad \forall_{i} \; x_i \in \{0,1\}.
\end{eqnarray}
For the $2$-rainbow domination problem it would be interesting to 
know on what conditions on the adjacency matrix $A$, the matrix $B$ is (totally) 
balanced.  
\end{remark}
                 
\section{2-Rainbow domination of permutation graphs}
\label{appendix permutation}

Permutation graphs properly contain the class of cographs. 
They can be defined as the intersection graphs of a set of 
straight line segments that have their endpoints on two parallel 
lines. 
In other words, a graph $G$ is a permutation graph whenever   
both $G$ and $\Bar{G}$ are comparability graphs~\cite{kn:dushnik,kn:klavik}. 

\bigskip 

We can use a technique similar to that used for interval graphs to show 
that $2$-rainbow domination is polynomial for permutation graphs. 

An intersection model of straight line segments for a 
permutation graph, as described above,  
is called a permutation diagram. 

\begin{definition}
Consider a permutation diagram for a permutation graph $G$. 
A scanline is a straight line segment with endpoints on the two 
parallel lines of which neither of the endpoints coincides 
with any endpoint of a line segment of the permutation diagram. 
\end{definition}

\bigskip 

We proceed as in Section~\ref{section interval}. 
Consider a $2$-rainbow function $f$ for a graph $G=(V,E)$. 
For a 
set $S \subseteq V$ we write, $f(S)=\cup_{x \in S} f(x)$. 
  
\begin{lemma}
\label{lm perm1}
Let $G=(V,E)$ be a permutation graph and let $f$ be a $2$-rainbow 
function with $\sum|f(x)|=\gamma_{r2}(G)$. 
Consider a permutation diagram with two parallel horizontal 
lines. Let $s$ be a scanline. There are at most 4 line segments 
that cross $s$ with function value $\{1,2\}$. 
\end{lemma}
\begin{proof}
Assume that there are at least 5 such line segments. Take 4 of them, 
each one with the endpoint as far as possible from the endpoint of $s$ 
on the top and bottom line. Any line segment crossing the fifth, crosses 
also at least one of these four. This shows that the function value of the 
fifth line can be replaced by $\es$.  This contradicts the optimality 
of $f$. 

\medskip 

\noindent
This proves the lemma. 
\qed\end{proof}

\begin{lemma}
\label{lm perm2}
For every scanline $s$ there are at most 8 line segments that cross it 
and that have function values not $\es$.        
\end{lemma}
\begin{proof}
The proof is similar to the proof of Lemma~\ref{lm perm1}. 
Consider the 8 line segments that cross $s$, 
and with function values that are nonempty 
subsets of $\{1,2\}$, and of which the union is pairwise $\{1,2\}$, 
and that are furthest away from the endpoints of $s$. 

\medskip 

\noindent
Any other line segment crossing $s$ with function value $\neq \es$ 
would contradict the optimality of $f$. 
\qed\end{proof}

\bigskip 

\begin{theorem}
\label{thm perm}
There exists a polynomial algorithm to compute the 
$2$-rainbow domination number for permutation graphs. 
\end{theorem}
\begin{proof}
Consider a dynamic programming algorithm, similar to that 
described for interval graphs in Theorem~\ref{thm int}, that proceeds 
by moving a scanline from left to right through the permutation diagram. 
\qed\end{proof}

\bigskip 

\begin{remark}
Obviously, a similar technique shows that the weak $\{2\}$-domination number 
can be computed in polynomial-time for permutation graphs. 
Perhaps it is interesting to ask the question whether these two domination 
numbers are equal for the class of permutation graphs. 
\end{remark}
  
\section{Weak $\{k\}$-L-domination on complete bipartite graphs}
\label{appendix CB}

For integers $1 \leq j \leq k$, a $(j,k)$-dominating function on a graph $G$ 
is a 
function $g: V(G) \rightarrow \{0,\dots,j\}$ such that for every vertex $x$, 
$g(N[x]) \geq k$~\cite{kn:rubalcaba,kn:rubalcaba2}. The $(j,k)$-domination 
number $\gamma_{j,k}(G)$ is the minimal cost of a $(j,k)$-dominating function.  

\begin{theorem}
Let $G$ be trivially perfect. There exists a linear-time algorithm 
to compute $\gamma_{j,k}(G)$. 
\end{theorem}
\begin{proof}
Let $T$ be a tree model for $G$. 
It is easy to check that the following procedure solves the problem. 
Color the vertices of the first $\lfloor \frac{k}{j} \rfloor$ $\BFS$-levels 
of the tree $T$ with color $j$, and color the vertices in the next level with the remainder 
$k-\lfloor \frac{k}{j} \rfloor \cdot j = k \bmod j$. 
\qed\end{proof}

\bigskip 
  
We show that the weak $\{k\}$-L-domination problem can be solved in linear time 
on complete bipartite graphs. 
Let $G$ be complete bipartite, with color classes $V$ and $V^{\prime}$. 
Let $L$ be a $\{k\}$-assignment, that is, $L$ assigns a pair $(a_x,b_x)$ 
of numbers from $\{0,\dots,k\}$ to every vertex $x$. 
For simplicity we may assume that, for all vertices $x$, 
\[a_x=0.\] 
For simplicity we also assume that $|V|=|V^{\prime}|=n$. 
Denote the $b$-labels of vertices in $V$ by $b(1),\dots,b(n)$ and 
denote the 
$b$-labels of vertices in $V^{\prime}$ by $b^{\prime}(1),\dots,b^{\prime}(n)$. 
We assume that these are ordered such that 
\[b(1) \; \geq\; \dots \;\geq\; b(n) \quad\text{and}\quad 
b^{\prime}(1) \;\geq\; \dots\;\geq\; b^{\prime}(n).\] 
Then the weak $\{k\}$-L-domination problem can be formulated as follows. 
Let $b(n+1)=b^{\prime}(n+1)=0$.  
\begin{gather*}
\min \; x+y \\
\text{subject to}\quad x \geq b^{\prime}(y+1) \quad \text{and}\quad y \geq b(x+1).
\end{gather*}

\bigskip 

\begin{theorem}
The weak $\{k\}$-L-domination problem can be solved in linear time 
on complete bipartite graphs. 
\end{theorem}
\begin{proof}
Let, for $x \in \{0,\dots,n\}$,  
\begin{eqnarray*}
y^1(x)&=&\min\;\{\;y\;|\; y \in \{0,\dots,n\} \quad\text{and}\quad 
y \geq b(x+1)\;\}=b(x+1)\\
y^2(x)&=& \min \; \{\;y\;|\; y \in \{0,\dots,n\} \quad\text{and}\quad 
x \geq b^{\prime}(y+1)\;\}.
\end{eqnarray*}
Let 
\[m(x)=\max\;\{\;y^1(x),\;y^2(x)\;\}.\] 
Then the solution to the weak $\{k\}$-L-domination problem is 
\[\min \;\{\;x+m(x)\;|\; x \in \{0,\dots,n\}\;\}.\] 
It is easy to check that the values $y^1(x)$ and $y^2(x)$ can be 
computed in, overall, linear time. 
This proves the theorem.
\qed\end{proof}

\end{document}